\documentclass[11pt,3p]{elsarticle}

\usepackage{amsfonts}
\usepackage{amsthm}
\usepackage{amsmath}
\usepackage{times}
\usepackage{lineno}
\usepackage{multirow,array}

\renewcommand{\epsilon}{\varepsilon}
\newcommand{\nats}{{\mathbb N}}
\def\A{\mathbb{A}}
\def\AP{\textit{AP\,}}
\def\d{\textit{\underbar{d}\,}}

\newtheorem{theorem}{Theorem}
\newtheorem{proposition}[theorem]{Proposition}
\newtheorem{lemma}[theorem]{Lemma}
\newtheorem{corollary}[theorem]{Corollary}
\newtheorem{definition}[theorem]{Definition}

\begin{document}

\sloppy


\begin{frontmatter}

\title{Anti-Powers in Infinite Words
}

\author[label1]{Gabriele Fici\corref{cor1}}
\ead{gabriele.fici@unipa.it}

\author[label1]{Antonio Restivo}
\ead{antonio.restivo@unipa.it}

\author[label2]{Manuel Silva}
\ead{mnasilva@gmail.com}

\author[label3]{Luca Q. Zamboni}
\ead{zamboni@math.univ-lyon1.fr}

\address[label1]{Dipartimento di Matematica e Informatica, Universit\`a di Palermo\\Via Archirafi 34, 90123 Palermo, Italy}
\address[label2]{Faculdade de Ci\^encias e Tecnologia,
Universidade Nova de Lisboa\\  Quinta da Torre, Campus Universit\'ario, 2829-516 Caparica, Portugal}
\address[label3]{Univ Lyon, Universit\'e Claude Bernard Lyon 1, CNRS UMR 5208, Institut Camille Jordan\\ 43 blvd. du 11 novembre 1918, F-69622 Villeurbanne cedex, France}

\cortext[cor1]{Corresponding author.}

\journal{Journal of Combinatorial Theory, Series A}

\begin{abstract}
In combinatorics of words, a concatenation of $k$ consecutive equal blocks is called a power of order $k$. In this paper we take a different point of view and define an anti-power of order $k$ as a concatenation of $k$ consecutive pairwise distinct blocks of the same length. As a main result, we show that every infinite word contains powers of any order or anti-powers of any order. That is, the existence of powers or anti-powers is an unavoidable regularity. Indeed, we prove a stronger result, which relates the density of anti-powers to the existence of a factor that occurs with arbitrary exponent. 
As a consequence, we show that in every aperiodic uniformly recurrent word, anti-powers of every order begin at every position.
We further show that every infinite word avoiding anti-powers of order $3$ is ultimately periodic, while there exist aperiodic words avoiding anti-powers of order $4$. We also show that there exist aperiodic recurrent words avoiding anti-powers of order $6$.
\end{abstract}


\end{frontmatter}

\section{Introduction}
Ramsey theory is an old and important area of combinatorics. Since the original result of 1930 by F.~Ramsey \cite{Ram30} research developed in several directions, but the crux of the matter remains the study of regularities that  must arise in large combinatorial structures. These kinds of regularities classically concern substructures formed by \emph{all-equal} elements (e.g., a monochromatic clique in an edge-colored graph).

However, at the beginning of the 70s, Erd{\H{o}}s, Simonovits and T. S{\'o}s \cite{ESS75} started the study of \emph{anti-Ramsey} theory, that is, the study of regularities that concern \emph{all-distinct} objects (e.g., a subgraph of an edge-colored graph in which all edges have different colors, often called a \emph{rainbow} --- see \cite{Fujita2010} for a survey).

In combinatorics on words, Ramsey theory found applications through some important results stating the existence of unavoidable regularities. Formally, an \emph{unavoidable regularity} is a property $P$ such that it is not possible to construct arbitrarily long words not satisfying $P$ (cf.~\cite{DelVa99}). Most of the main results about unavoidable regularities in words were originally stated in other areas of combinatorics, e.g., the Ramsey, van der Waerden and  Shirshov theorems (see~\cite{DelVa99,LothI,LothII} for further details). All these theorems, however, establish the existence, in every sufficiently long word, of regular substructures. In this paper, we give an anti-Ramsey result in the context of combinatorics on words.

Regularities in words are often associated with \emph{repetitions}, also called \emph{powers}. A power of order $k$ is a concatenation of $k$ identical copies of the same word. The most simple power is a square (a power of order $2$). As noted by Thue in 1906 \cite{Th06}, every sufficiently long binary word must contain a square, but there exist arbitrarily long words over a $3$-letter alphabet avoiding squares, that is, not containing any square as a block of contiguous letters (in terms of combinatorics on words, a \emph{factor}). This shows that the avoidability of powers depends on the alphabet size.

In this paper we introduce the notion of an anti-power. An {anti-power of order $k$}, or simply a \emph{$k$-anti-power}, is a concatenation of $k$ consecutive pairwise distinct words of the same length.  E.g., $aabaaabbbaba$ is a $4$-anti-power. A simple computation shows that there are in general much more anti-powers than powers for a fixed length and a fixed order; yet there are much less anti-powers than possible factors of the same given length. 

 Let us consider as an example the well-known Thue-Morse word \[t=0110100110010110100101100110100110010110011010\cdots\]
Starting from $n=0$, the $n$-th term is given by the parity of the number of $1$s in the binary expansion of $n$. 
The Thue-Morse word does not contain \emph{overlaps}, i.e., factors of the form $awawa$ for a letter $a$ and a word $w$.  
In particular, the Thue-Morse word does not contain $3$-powers.

\begin{table}
\centering  
\begin{small}
\begin{raggedright}
\begin{tabular}{c *{30}{@{\hspace{1.3mm}}c}}
$k$\hspace{2mm} & 3 & 4 & 5 & 6 & 7 & 8 & 9 & 10 & 11 & 12 & 13 & 14 & 15 & 16 & 17 & 18 & 19 & 20 & \  & 30 & \  &  50 & \  &  100
\\
\hline \\
length\hspace{2mm}  &15 & 20 & 25 & 30 & 77 & 88 & 99 & 110 & 121 & 132 & 143 & 154 & 195 & 208 & 221 & 234 & 247 & 260 & & 870 & & 2450 & & 9700
\\
\hline \\
ratio\hspace{2mm}  &5 & 5 & 5 & 5 & 11 & 11 & 11 & 11 & 11 & 11 & 11 & 11 & 13 & 13 & 13 & 13 & 13 & 13 
&& 29 && 49 && 97\\
\hline \rule[0pt]{0pt}{12pt}
\end{tabular}
\end{raggedright}\caption{\label{tab:tm} The first few values of the sequence of  lengths of the shortest prefixes of the Thue-Morse word that are $k$-anti-powers.}
\end{small}
\end{table}

The shortest prefix of the Thue-Morse word that is a $2$-antipower is $01$. The shortest prefix that is a $3$-anti-power is $01101\cdot 00110\cdot 01011$, of length $15$. One can verify that the shortest $4$-anti-power prefix has length $20$. The first few lengths of the shortest prefixes of $t$ that are $k$-anti-powers for different values of $k$ are displayed in Table~\ref{tab:tm}. A natural question is therefore the following: Given an integer $k>1$, is it always possible to find a prefix of $t$ that is a $k$-anti-power? In this paper we answer this question in the affirmative. Actually, we prove a much stronger result. Indeed, we prove that the existence of powers of any order or anti-powers of any order is an unavoidable regularity:

\begin{theorem}\label{thm:main}
Every infinite word contains powers of any order or anti-powers of any order.
\end{theorem}

Note that we do not make any hypothesis on the alphabet size. Actually, we prove a stronger result, from which Theorem \ref{thm:main} follows. Given an infinite word $x$, we prove that if for some integer $k$ the lower density of the set of lengths $n$ for which the prefix of $x$ of length $kn$ is a $k$-anti-power is smaller than one, then there exists a word (whose length depends on $k$) that occurs in $x$ with arbitrary exponent (Theorem \ref{prop:luca}). This implies that if an infinite word $x$ has the property that each of its factors appears with bounded exponent (in the terminology of combinatorics on words, an \emph{$\omega$-power-free} word), then 
$x$ must begin in an anti-power of order $k$ for every choice of $k.$
In particular, since a uniformly recurrent word is either purely periodic or  $\omega$-power-free, this property holds for every aperiodic uniformly recurrent word, as for example the Thue-Morse word or any Sturmian word\footnote{Sturmian words are aperiodic words of minimal factor complexity. They are very well studied objects in combinatorics on words (see for instance \cite{LothII}).}.

In the second part of the paper, we focus on the avoidability of anti-powers. We show that any infinite word avoiding $3$-anti-powers is ultimately periodic, and that there exist aperiodic words avoiding $4$-anti-powers. We also show that there exist aperiodic \emph{recurrent} words avoiding $6$-anti-powers. We leave it as an open question to determine whether there exist aperiodic recurrent words avoiding $4$-anti-powers or $5$-anti-powers.

We conclude with final considerations and discuss open problems and further possible directions of investigation.

\section{Preliminaries}

Let $\nats=\{1,2,3,\ldots\}$ be the set of natural numbers.  Let $\A$ be a (possibly infinite) non-empty set, called the \emph{alphabet}, whose elements are called \emph{letters}.  
A \emph{word} over $\A$ is a finite or infinite sequence of letters from $\A$. The \emph{length} $|u|$ of a finite word $u$ is the number of its letters. The \emph{empty word} $\epsilon$ is the word of length $0$. We let $\A^+$ denote the set of all finite words of positive length over $\A$, and $\A^\nats$  the set of all infinite words over $\A$, that is, the set of all maps from $\nats$ to $\A$. Given a finite word $u$, we let $u^\omega$ denote the infinite word $uuu\cdots$ obtained by concatenating an infinite number of copies of $u$.

Given a finite or infinite word $x$, we say that a word $u$ is a \emph{factor} of $x$ if $x=vuy$ for some (possibly empty) words $v$ and $y$. We say that $u$ is a \emph{prefix} (resp.~\emph{suffix}) of $x$ if $x=uy$ (resp.~$x=yu$) for some word $y$. We say that a word $u\neq x$ is a \emph{border} of $x$ if $u$ is both a prefix and a suffix of $x$.

An infinite word $x$ is \emph{purely periodic} if there exists a positive integer $p$ such that the letters occurring at positions $i$ and $j$ coincide whenever $i=j\mod p$. Equivalently, $x$ is purely periodic if and only if $x=u^\omega$ for some word $u$ of length $p$. An infinite word $x$ is \emph{ultimately periodic} if $x=uy$ for a finite word $u$ and a purely periodic word $y$. An infinite word is \emph{aperiodic} if it is not ultimately periodic.

An infinite word $x$ is said to be \emph{recurrent} if every finite factor of $x$ occurs in $x$ infinitely often. Equivalently, $x$ is recurrent if and only if every finite prefix of $x$ has a second occurrence as a factor. An infinite word $x$ is said to be \emph{uniformly recurrent} if every finite factor of $x$ occurs syndetically  (that is, it occurs infinitely often and with bounded gaps). Equivalently, $x$ is uniformly recurrent if and only if for every finite factor $u$ of $x$ there exits an integer $m$ such that $u$ occurs in every factor of $x$ of length $m$. 

An infinite word $x$ is said to be \emph{$k$-power-free} for some integer $k>1$ if for every finite factor $u$ of $x$, one has that $u^k$ is not a factor of $x$. An infinite word $x$ is said to be \emph{$\omega$-power-free} if for every finite factor $u$ of $x$ there exists a positive integer $l$ such that $u^l$ is not a factor of $x$. Of course, if a word is $k$-power-free for some integer $k$, then it is $\omega$-power-free, but the converse is not always true.

An important relationship between uniformly recurrent and $\omega$-power-free words is the following (see for instance \cite{DelVa99}):

\begin{theorem}\label{thm:del}
 Every uniformly recurrent word is either purely periodic or $\omega$-power-free.
\end{theorem}

\section{Unavoidability of powers or anti-powers}

In order to state our main result, we need to introduce some definitions.

Let $x$ be an infinite word and $k\in \nats$. We set
\[P(x,k)=\{m\in \nats \mid \mbox{ the prefix of $x$ of length $km$ is a $k$-power} \}\]
Analogously, we set
\[\AP(x,k)=\{m\in \nats \mid \mbox{ the prefix of $x$ of length $km$ is a $k$-anti-power} \}\]

Note that $P(x,1)=\AP(x,1)=\nats$ and that $P(x,k)\cap \AP(x,k)=\emptyset$ for every $k\geq 2.$ 
For example, if $x=01^\omega,$ we have $P(x,k)= \AP(x,k)=\emptyset$ for every $k\geq 3.$ 

Recall that for any subset $X\subseteq \nats$, the \emph{lower density} of $X$ is defined by 
\[\d(X)=\liminf_{n\rightarrow \infty} \frac{|X \cap \{1,2,\ldots ,n\}|}{n}.\]
Note that if $X$ is finite, then $\d(X)=0$. Moreover, if $\d(X)< 1/t$ for some integer $t>0$, then  there exist infinitely many integers $m$ such that $\{m,m+1,\ldots ,m+t-1\}\subset  \nats \setminus X.$

We are now going to prove our main result (Theorem \ref{prop:luca}). 
We first need a technical lemma.

\begin{lemma}\label{lem:antonio}
Let $v$ be a border of a word $w$ and let $u$ be the word such that $w=uv$. If $l$ is an integer such that $|w|\geq l|u|$, then $u^l$ is a prefix of $w$.
\end{lemma}
\begin{proof}
By induction on $l$. For $l=1$ the statement trivially holds. Suppose $l>1$. Since $u$ is shorter than $v$ and both are prefixes of $w$, we have that $u$ is a prefix of $v$. Let us write $v=uv'$. Then $w=uuv'$ and $v'$ is a border of $v$. Since $|v|=|w|-|u|\geq (l-1)|u|$, we can apply the induction hypothesis and derive that $u^{l-1}$ is a prefix of $v$, whence $u^l$ is a prefix of $w$.
\end{proof}

\begin{theorem}\label{prop:luca} 
Let $x$ be an infinite word. Suppose that \[\d(\AP(x,k))<\left(1+{k\choose 2}\right)^{-1}=\frac{2}{2+k(k-1)}\] for some $k\in \nats.$ Then there exists $u\in \A^+$ with
$|u|\leq (k-1){k\choose 2}$ such that $u^l$ is a factor of $x$ for every $l\geq 1.$ 
\end{theorem}

\begin{proof}  
Fix $k$ such that $\d(\AP(x,k))<(1+{k\choose 2})^{-1}.$ Since $\AP(x,1)=\nats,$ and the lower density of $\nats$ is $1$, we have $k\geq 2.$ We set $M=(k-1){k\choose 2}.$ 
We have to show there exists $u\in \A^+$ with $|u|\leq M$ such that $u^l$ is a factor of $x$ for every $l\geq 1.$  By the pigeonhole principle, it suffices to show that for every  $l\in \nats$ there exists $u\in \A^+$ with $|u|\leq M$ such that $u^l$ is a factor of $x.$ 

So, let us fix $l\in \nats, $ and set $N=(l+1)M.$ 
Since $\d(\AP(x,k))<(1+{k\choose 2})^{-1},$  there exists an integer  $m> N$ such that $\{m,m+1,\ldots ,m+{k\choose 2}\}\subset  \nats \setminus \AP(x,k).$

For every $j$ and $r$ such that $0\leq j\leq k-1$ and $m\leq r\leq m+{k\choose 2}$, set
\[U_{j,r}=x_{jr+1}x_{jr+2}\cdots x_{(j+1)r},\] so that $|U_{j,r}|=r$ and $U_{0,r},U_{1,r},\ldots , U_{k-1,r}$ are the first $k$ consecutive blocks of $x$ of length $r.$ Thus for each $m\leq r\leq m+{k\choose 2}$ there exist $i$ and $j$, with $0\leq i<j\leq k-1$, such that $U_{i,r}=U_{j,r}.$  By the pigeonhole principle, there exist $r$ and $s$, with $m\leq r<s\leq m+{k\choose 2}$, and $i$ and $j$, with $0\leq i<j\leq k-1$, such that $U_{i,r}=U_{j,r}$ and $U_{i,s}=U_{j,s}$.
\begin{figure}
\begin{center}
 \includegraphics[scale=0.38]{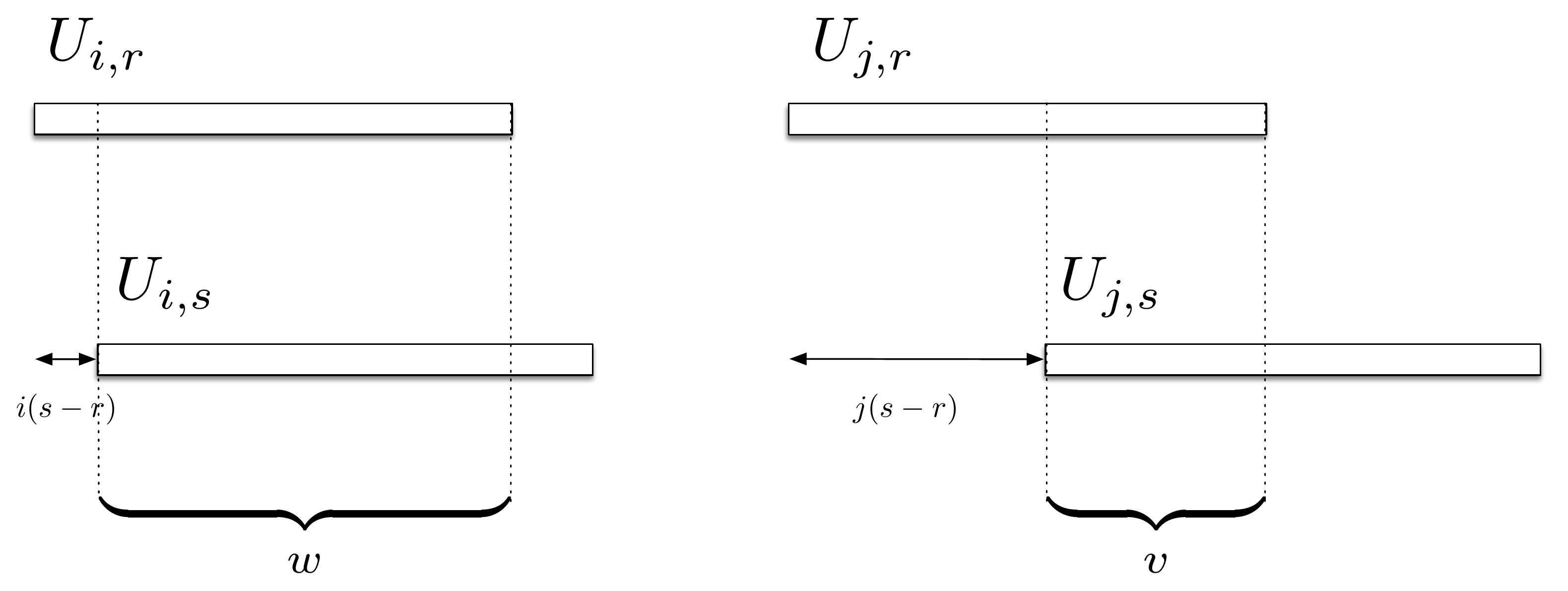}
 \caption{The proof of Theorem \ref{prop:luca}.}
 \label{fig:thm}
\end{center}
\end{figure}

Notice that $(i+1)r>is+1$ and $(j+1)r>js+1$.

Let us now set $w=x_{is+1}x_{is+2}\cdots x_{(i+1)r}$ and $v=x_{js+1}x_{js+2}\cdots x_{(j+1)r}$ (see Figure \ref{fig:thm}). We have \[|v|=(j+1)r-js<(i+1)r-is=|w|,\] whence $v$ is a border of $w.$ Writing $w=uv$, we have \[1\leq |u|=|w|-|v|=(j-i)(s-r)\leq (k-1) {k\choose 2}=M,\] and \[|w|>|v|=r-j(s-r)\geq m-(k-1){k\choose 2}=m-M>N-M=lM.\]
Thus, $|w|>l|u|$ and, by Lemma~\ref{lem:antonio}, $u^l$ is a prefix of $w$, and therefore $u^l$  is a factor of $x.$
\end{proof}

A  consequence of Theorem \ref{prop:luca} is the following.

\begin{corollary} Let $x$ be a uniformly recurrent word.  Then one, and only one, of the following cases holds:
\begin{enumerate}
 \item for every $k\in \nats$ \begin{equation}\d(\AP(x,k))\geq\frac{2}{2+k(k-1)}\ \end{equation} (and in this case $x$ is aperiodic);
 \item   there exists $r>0$ such that for every $k\in \nats$ 
\begin{equation}\d(P(x,k))\geq r\end{equation} (and in this case $x$ is periodic).
\end{enumerate}  
\end{corollary}

  \begin{proof} According to Theorem~\ref{prop:luca}, if $(1)$ does not hold for some $k'\in \nats,$ then $x=u^\omega$ for some $u$ with $1\leq |u|\leq (k'-1){k'\choose 2}.$ Whence $n|u|\in P(x,k)$ for each $n,k\in \nats.$ 
 The result now follows by setting $r=1/|u|.$ 
 \end{proof}
 
 Note that the $\d(P(x,k))\geq r$ for every $k$ is stronger than just $\d(P(x,k))>0$.
Conversely, if $\d(P(x,k))>0$ for some $k\geq 2$, then $\d(P(x,2)) >0$, and from this it is immediate to see that $x$ is periodic. 

In fact, something stronger is true: following the notation in the proof of Theorem \ref{prop:luca}, if there exists $j\geq 1$ such that $\d\{r\mid U_{0,r}=U_{j,r}\}>0$, then $x$ is periodic (indeed, by the pigeonhole principle, our assumption implies the existence of a positive integer $s$ and infinitely many positive integers $r$ such that
$U_{0,r} =U_{j,r}$ and $U_{0,r+s}=U_{j,r+s}$; again by the pigeonhole principle, this in turn implies that there exists a nonempty factor $v$ of $x$ such that $v^n$ is a factor of $x$ for every positive integer $n$, which, as $x$ is uniformly recurrent, implies that $x$ is periodic). Note that $\d(P(x,2))>0$ is a special case of this assumption (when $j=1$).
 
Another direct consequence of Theorem \ref{prop:luca} is the following.

\begin{theorem}\label{thm:main2}
Let $x$ be a $\omega$-power-free word. Then for every $k>1$ there is an occurrence of an anti-power of order $k$ starting at every position of $x$.
\end{theorem}

\begin{proof}
 Suppose that there exists a positive integer $k$ and a suffix $x'$ of $x$ such that no prefix of $x'$ is a $k$-anti-power. Then $\AP(x',k)=\emptyset$, whence $\d(\AP(x',k))=0$. By Theorem  \ref{prop:luca}, there exists a factor $u$ of $x'$ such that $u^l$ is a factor of $x'$ for every $l\geq 1$, hence $x$ is not $\omega$-power-free.
\end{proof}

From Theorems \ref{thm:del} and \ref{thm:main2}, we derive the following corollary.

\begin{corollary}\label{cor:ur}
 Let $x$ be a uniformly recurrent aperiodic word. Then for every $k>1$ there is an occurrence of an anti-power of order $k$ starting at every position of $x$.
\end{corollary}

\section{Avoiding anti-powers}

In this section we deal with the avoidability of anti-powers. 

\begin{definition}
Given $k>1$, we say that an infinite word $x$ avoids $k$-anti-powers if no factor of $x$ is a $k$-anti-power. That is, among any $k$ consecutive blocks of the same length in $x$, at least two of them are equal.
We say that an infinite word $x$ avoids anti-powers if $x$ avoids $k$-anti-powers for some $k$. \end{definition}
 
Periodic words avoid anti-powers, the period length being an upper bound for the maximal number of distinct consecutive  blocks of the same length. In the following, we discuss the avoidability of anti-powers for aperiodic words. By Corollary \ref{cor:ur}, if an aperiodic word avoids anti-powers, then it  cannot be uniformly recurrent.
 
Of course, any word containing at least two different letters cannot avoid $2$-anti-powers. For $3$-anti-powers, we have the following result.

\begin{lemma}
Let $x$ be an infinite word. If $x$ avoids $3$-anti-powers, then $x$ is a binary word.
\end{lemma}

\begin{proof}
Suppose $x$ avoids $3$-anti-powers and contains three different letters. Then there is a factor of $x$ of the form $u=ab^nc$
with $n\geq 1$ and $a,b,c$ distinct letters. We will extend this factor to the right and force a
3-anti-power for every $n$. For $n=1$, the word $abc$ is already an anti-power. Take now $n=2$. To avoid
3-anti-powers, $abbc$ can only be extended to $abbcb$. In the next step,
the only option is $abbcbc$, and after that $abbcbcb$. But now, the word
$abbcbcbyy'$ contains a 3-anti-power for every letters $y,y'$. Suppose now $u=ab^nc$ with $n \geq 3$. If $n$ is odd, let $m=(n-1)/2$
and note that $u$ can be factored as $ab^m\cdot b^{m+1}\cdot c$, so that $u$ 
will be extended to the right to a 3-anti-power of length $3(m+1)$. If $n$ is even, $u$ can be factored as $u=ab^m\cdot b^{m+1}\cdot bc$, so that again $u$ will be extended to the right to a 3-anti-power of length $3(m+1)$.
\end{proof}

Hence, in what follows we will suppose that $x$ is an infinite word over the binary alphabet $\A=\{0,1\}$. 

\begin{proposition}\label{prop:man}
Let $x$ be an infinite word. If $x$ avoids $3$-anti-powers, then it cannot contain a factor of the form $10^n1$ or $01^n0$ with $n>1$.
\end{proposition}

\begin{proof}
Suppose that $x$ contains a factor of the form $u=10^n1$ with $n>1$ (the other situation is symmetric). The cases $n=2,3,4,5$ can be checked by computer, so let us suppose $n\geq 6$.

Suppose first $n$ even, and write $n=2m$. Since $u=10^{m-1}\cdot 0^m \cdot 01$,  any extension of $u$ to the right  will produce a $3$-anti-power of length $3m$. If $n$ is odd, $n=2m+1$, then we can write $u=10^{m-1}\cdot 0^{m} \cdot 001$, so that any extension of $u$ to the right  will produce a $3$-anti-power of length $3m$.
\end{proof}

\begin{corollary} 
Let $x$ be an infinite word avoiding $3$-anti-powers. Then $x$ is ultimately periodic.
\end{corollary}

\begin{proposition}
There exist aperiodic  words avoiding $4$-anti-powers.
\end{proposition}

\begin{proof}
We exhibit an example of an aperiodic word avoiding $4$-anti-powers. 
 Let $(\alpha_i)_{i\geq 1}$ be an increasing sequence of positive integers with $\alpha_{i+1}\geq 5\alpha_i$ for each $i\geq 1.$ Let $x \in \{0,1\}^\nats$ be defined by $x_n=1$ if $n=\alpha_i$ for some $i\geq 1,$ and $x_n=0$ otherwise. Clearly $x$ is aperiodic. Moreover, given $m\geq 0$ and $n\in \nats, $ if $|x_{m+1}x_{m+2}\cdots x_{m+n}|_1\geq 2,$ then for some $i\geq 1$ \[m+1\leq \alpha_i<5\alpha_i\leq \alpha_{i+1}\leq m+n\] and hence
 $n>4\alpha_i\geq 4(m+1)$ whence $m+1<n/4.$ We claim that $x$ avoids $4$-anti-powers. In fact, suppose to the contrary that for some $m\geq 0$ and $n\in \nats$ we have $x_{m+1}\cdots x_{m+n},$ $x_{m+n+1}\cdots x_{m+2n},$ $x_{m+2n+1}\cdots x_{m+3n},$ and $x_{m+3n+1}\cdots x_{m+4n}$ are pairwise distinct. Then at least three of the four blocks must contain an occurrence of $1.$  Thus $|x_{m+n+1}\cdots  x_{m+4n}|_1\geq 2$ from which it follows that $m+n+1<3n/4$ and hence $m+1<0,$ a contradiction.
\end{proof}

The word in the previous proposition is not recurrent. It is natural to ask whether there exist recurrent words avoiding $4$-anti-powers. We do not know the answer. However, we can state the following result.

\begin{proposition}
There exist aperiodic recurrent words avoiding $6$-anti-powers.  
\end{proposition}

\begin{proof}
We exhibit an example of an aperiodic recurrent word avoiding $6$-anti-powers.  Let $w_0=0$ and $w_n=w_{n-1}1^{3|w_{n-1}|}w_{n-1}$ for $n>0$. Let $w$ be the infinite word obtained as the limit of the sequence of words $(w_n)_{n\geq 1}$. Then clearly $w$ is recurrent.  

Notice that each occurrence of $w_n$ in $w$, except the first one, is preceded by $1^{3|w_n|}$.

Let $v=v_1v_2\cdots v_6$ be a non-empty factor of $w$ of length $6\ell$ for some integer $\ell$. Let $n$ be the largest integer such that $|w_n|=5^n< 2\ell$. 
We consider an occurrence of $v$ beginning at a position that follows the second occurrence of $w_n$ in $w$. We can do this since $w$ is recurrent.
By the hypothesis on $n$, no $v_i$ can intersect two occurrences of $w_n$.

Suppose first that for some $i$, $v_i$ is contained as a factor in $w_n$. By the hypothesis on $n$, neither $v_{i-1}v_i$ nor $v_iv_{i+1}$ is contained in $w_n$. Since $w_n$ is preceded and followed by $1^{3|w_n|}$, either $v_{i-3}$ and $v_{i-2}$ (if $i\geq 4$) or $v_{i+2}$ and $v_{i+3}$ (if $i< 4$) are both equal to $1^\ell$, so that $v$ cannot be an anti-power. 
 
  If instead no $v_i$ is contained as a factor in $w_n$, then one of the following cases must hold:
  
  \textit{i}) There is an occurrence of $w_n$ intersecting $v_{i}$ and the next occurrence of $w_n$ intersects $v_{i+1}$. In this case, either $v_{i-3}$ and $v_{i-2}$ or $v_{i+3}$ and $v_{i+4}$ are both equal to $1^\ell$.
  
  \textit{ii}) There is an occurrence of $w_n$ intersecting $v_{i}$ and the next occurrence of $w_n$ intersects $v_{i+2}$, so that $v_{i+1}=1^\ell$. In this case, either $v_{i-2}$ or $v_{i+4}$ must be equal to $1^\ell$.
  
  \textit{iii}) There are two consecutive blocks $v_i$, $v_{i+1}$ both equal to $1^\ell$.
  
In all cases, $v$ cannot be a $6$-anti-power. 
\end{proof}

\section{Conclusions and open problems}

We proved that every infinite word contains powers of any order or anti-powers of any order, that is, the existence of powers or anti-powers is an unavoidable regularity. This result can also be stated in the following finite version.

\begin{theorem}\label{thm:mainF}
For all integers $l>1$ and $k>1$ there exists $N=N(l,k)$ such that every word of length $N$ contains an $l$-power or a $k$-anti-power. Furthermore, for $k>2$, one has $k^2-1\leq N(k,k) \leq k^3{k\choose 2}$.
\end{theorem}

The upper bound follows from the proof of Theorem \ref{prop:luca}. Indeed, let $x$ be any word of length $k^3{k\choose 2}$. 
Set $m=(k+1)(k-1){k \choose 2}$. Consider for every $r$ such that $m\leq r \leq m+{k \choose 2}$, the first $k$ consecutive blocks of length $r$ in $x$: $U_{0,r},U_{1,r},\ldots, U_{k-1,r}$. Reasoning as in the proof of Theorem \ref{prop:luca}, if $x$ does not contain $k$-anti-powers, then by the pigeonhole principle there exist $r$ and $s$, with $m\leq r<s\leq m+{k\choose 2}$, and $i$ and $j$, with $0\leq i<j\leq k-1$, such that $U_{i,r}=U_{j,r}$ and $U_{i,s}=U_{j,s}$. We then set $w=x_{is+1}x_{is+2}\cdots x_{(i+1)r}$ and $v=x_{js+1}x_{js+2}\cdots x_{(j+1)r}$ (see Figure \ref{fig:thm}), and we have that $v$ is a border of $w$. Writing $w=uv$, we have $|u|\leq (k-1){k\choose 2}$. Now, since \[|w|>|v|\geq m-(k-1){k\choose 2}=k(k-1){k\choose 2}\geq k|u|,\] we can apply Lemma~\ref{lem:antonio} and get that $u^k$ is a factor of $x$. 

Notice that bound $k^3{k\choose 2}$ on the length of $x$ is chosen to accomodate the $k$ blocks $U$, whose maximal length is $m+{k\choose 2}=k^2{k\choose 2}$.

As for the lower bound, for any $k>2$ the word $(0^{k-1}1)^{k-2}0^{k-2}10^{k-1}$ of length $k^2-2$ avoids both $k$-powers and $k$-anti-powers.

In the case of binary alphabet, it can be verified that $N(2,2)=2$,  $N(3,2)=3$, $N(2,3)=4$, $N(3,3)=9$, $N(4,3)=12$, $N(3,4)>16$, $N(4,4)>16$. We do not know how these numbers grow.
Moreover, the bounds on $N(l,k)$ given in Theorem~\ref{thm:mainF} can probably be improved by a deeper analysis of the function $N(l,k)$.

Concerning the avoidability of anti-powers, we proved that there exist words avoiding $4$ anti-powers and that there exist recurrent words avoiding $6$-anti-powers. A natural problem is therefore that of determining what is the least $k$ such that there exists a recurrent word avoiding $k$-anti-powers. 

Another possible direction of investigation is related to the possible lengths of anti-powers appearing in a word. For an aperiodic uniformly recurrent word $x$, define $ap(x,k)=\min (\AP(x,k))$, i.e., the minimum length $m$ for which the prefix of $x$ of length $km$ is a $k$-anti-power. The first values of this function for the Thue-Morse word are displayed in Table~\ref{tab:tm} (where the value of $ap(x,k)$ is the ratio between the length of the $k$-anti-power prefix and the order $k$). We wonder whether it is possible to link the behavior of the function $ap(x,k)$ to the combinatorics of the word $x$, at least for special classes of words. 
During the preparation of the final version of this paper, we became aware of an article by C. Defant showing interesting properties related to the anti-powers in the Thue-Morse word $t$, in particular showing that the  function $ap(t,k)$ grows linearly in $k$~\cite{Defant17}.

\section{Acknowledgements}

We thank Jeff Shallit for carefully reading a preliminary version of this paper. We also thank the anonymous referees for their valuable comments. 

\bibliographystyle{plain}

\end{document}